\documentclass[12pt,reqno,fleqn]{amsart}
\usepackage{latexsym}
\usepackage{amssymb}
\usepackage{amsxtra}
\usepackage{amsmath}
\usepackage{amsfonts}
\usepackage[dvips]{graphicx}
\usepackage{pstricks,pst-node}
\usepackage[metapost]{mfpic}
\usepackage{bmpsize}
\usepackage{mathrsfs}
\usepackage{enumerate}
\newtheorem{theorem}{Theorem}[section]

\makeatother       

\makeatletter      
\@addtoreset{equation}{section}
\makeatother       

\begin{document}

\title[]{The wronskian solution of the constrained discrete KP hierarchy}
\author{Maohua Li$^{1,2,3}$,   Jingsong
He$^1$$^*$} \dedicatory {
1. Department of Mathematics, Ningbo University, Ningbo, 315211 Zhejiang, China\\
2. School of Mathematical Sciences, USTC, Hefei, 230026 Anhui, China\\
3. Institute for Theoretical Physics, KU Leuven, Celestijnenlaan 200D, B-3001 Leuven, Belgium\\
limaohua@nbu.edu.cn\\
hejingsong@nbu.edu.cn
}
\date{}

\thanks{$^*$ Corresponding author}

\begin{abstract}
From the constrained discrete KP (cdKP) hierarchy,  the Ablowitz-Ladik lattice has been derived. By means of the gauge transformation, the Wronskian solution of the Ablowitz-Ladik lattice have been given. The $u_1$ of the cdKP hierarchy  is a Y-type soliton solution for odd times of the gauge transformation, but it becomes a dark-bright soliton solution for even times of the gauge transformation. The role of the discrete variable $n$ in the profile of the $u_1$ is discussed.
\end{abstract}

\maketitle
Keywords: {Constrained discrete KP hierarchy, Gauge transformation, Wronskian solution}

Mathematics Subject Classification (2000):  37K10, 37K40, 35Q51, 35Q55.

\section{Introduction}
In the past few years, lots of attention have been given to the study of
 Kadomtsev-Petviashvili (KP) hierarchy \cite{dkjm,dl1} in the field of integrable systems.
The Lax pairs, Hamiltonian structures, symmetries and conservation
laws, the $N$-soliton, tau function,  the gauge transformation,
reductions etc. of the KP hierarchy and its sub-hierarchies have
been discussed. There are several sub-hierarchies of the KP by considering different
reduction conditions on the Lax operator $L$. One of them is called
constrained KP (cKP) hierarchy \cite{kss,chengyiliyishenpla1991,
chengyiliyishenjmp1995} by setting the  Lax operator as  $L=\partial
+\sum _{i=1}^m \Phi_i \partial ^{-i}\Psi _i$.  The cKP hierarchy contains a large number of interesting
 soliton equations. The basic idea of this procedure is
so-called symmetry constraint \cite{kss,chengyiliyishenpla1991,
chengyiliyishenjmp1995}. The negative part of the Lax operator of the
constrained KP, i.e. $\sum _{i=1}^m \Phi_i \partial ^{-i}\Psi _i$,
is a generator \cite{dl1} of the additional symmetry \cite{os1} of the KP hierarchy.
And the additional symmetry of BKP hierarchy and CKP hierarchy have been given \cite{ tu07lmp,hetian07}.
Very recently, by a further modification of the additional flows,
 the additional symmetries of
 the constrained BKP and constrained CKP hierarchies are given in references
 \cite{tianhe2011,tu_shenJMP2011}.

It is well known that
a continuous integrable system  has a discrete analogue in general.
The famous 3-dimensional difference equation is known to provide a canonical integrable
discretization for most important types of soliton equations.
There are several different kinds of the discrete hierarchies including
differential-difference
 KP (dKP) hierarchy \cite{Kupershimidt,Iliev}, semi-discrete integrable systems, full discrete equations and so on. The differential-difference
 KP  hierarchy, defined by the difference operator
 $\Delta$, is one interesting object  of the discrete integrable systems. Note that,
 the additional symmetry of dKP hierarchy and it's Sato B\"acklund
 transformations  have been given in reference \cite{LiuS2}.
 Moreover, gauge transformation is one kind of powerful method to construct the
solutions of the integrable systems for both the continuous KP
hierarchy \cite{chau_cmp1992,oevel1993,oevelRMP93,wo1,nimmo,cst1,hlc2002,hlc2,hwc}
and the dKP hierarchy \cite{oe,LiuS}.
 It is discussed to reduce the gauge transformation of the dKP hierarchy
 to the constrained discrete KP(cdKP) hierarchy \cite{lmh20131}. And the algebraic structure of the additional symmetry of the cdKP hierarchy also has been found \cite{lmh2}, which is same for the cKP hierarchy \cite{hetian07}.

 A crucial observation \cite{Iliev} about the KP hierarchy and the dKP hierarchy is that
the $\tau$ function of the discrete KP hierarchy  can be constructed by  shift of the $t$ of the $\tau$ function of the continuous KP hierarchy.  It is an interesting question to find any other  difference among the two hierarchies.  In this direction, the correspondence between the solutions of discrete and continuous  hierarchy can be used to explore the difference between them.  In particular, a key step is  to demonstrate how the discrete variable $n$ affects the profile of the solutions of the dKP hierarchy.

  The purpose of this paper is to find the the correspondence between the solutions of the KP hierarchy and the dKP hierarchy by means of the multi-channel  gauge transformation. The paper is organized as follows. Some basic results  of the dKP hierarchy and the cdKP hierarchy are summarized in Section
\ref{section2}. The main theorem about the solution of cdKP hierarchy are give in Section \ref{section3}.  An example is give in section \ref{example}. We find that the odd kinds of gauge transformation of cdKP hierarchy can change to a new profile of solution of the cdKP hierarchy. Section
\ref{conclusion} is devoted to conclusions and discussions.

\section{the cdKP hierarchy} \label{section2}

Let $L$ be a general first-order pseudo difference operator(PDO)
\begin{equation} \label{laxoperatordkp}
L(n)=\Delta + \sum_{i=1}^{\infty} u_i(n)\Delta^{-i},
\end{equation}
 the  cdKP hierarchy \cite{lmh2} is defined by the following Lax equation
\begin{equation}\label{1cdKPlaxeq}
\frac{\partial L}{\partial t_l} = [B_l,L], B_l:=(L^l)_+, l=1,2,\cdots,
\end{equation}
associated with a constrained Lax operator
\begin{equation} \label{laxofcdkp}
 L^l_{-}= \sum_{i=1}^{m}q_i(t)\Delta^{-1}r_i(t),
\end{equation} which is  $m$-components Lax operator of the cdKP hierarchy. It has relation between the dynamical variables $q_i,r_i$ and $u_i$. Specially, $u_1=q_1\Lambda^{-1}(r_1)$, where $\Delta=\Lambda-I$.
The eigenfunction and adjoint eigenfunction {$q_i(t),r_i(t)$} are
important dynamical variables in the cdKP hierarchy.
It can be checked that the Lax equation (\ref{1cdKPlaxeq}) is consistent with the evolution
 equations of the eigenfunction (or adjoint eigenfunction)
\begin{eqnarray}\label{eigenfunction}
\begin{cases}
q_{i,t_m} = B_mq_i,\\
r_{i,t_m} = -B_m^*r_i,\quad  B_m=(L^m)_{+}, \forall m \in N.
\end{cases}
\end{eqnarray}
Therefore  the cdKP hierarchy in eq.(\ref{1cdKPlaxeq}) is well defined.

 From the Lax equation (\ref{1cdKPlaxeq}), we get the first nontrival $t_2$ flow equations of the cdKP hierarchy for $m=1, l=2$ as
\begin{eqnarray}\label{xuedinger}
\begin{cases}
q_{1,t_2}=\Delta^2 q_1+2q_1^2r_1=q_1(n+2)-2q_1(n+1)+q_1(n)+2q_1^2r_1,\\
r_{1,t_2}=-{\Delta^*}^2 r_1+2q_1r_1^2=r_1(n)-2r_1(n-1)+r_1(n-2)+2q_1(n)r_1(n)^2.
\end{cases}
\end{eqnarray}
It is nothing but the Ablowitz-Ladik lattice \cite{ablowitz1975}. It can be reduced to the discrete non-linear Schr\"{o}dinger (DNLS) equation \cite{ablowitz04} by letting $r_1=q_1^*$ and a scaling transformation $t_2=it_2$.

The Lax operator in
eq.(\ref{laxofcdkp}) can be generated by the dressing action
\begin{equation}
L=W \circ \Delta \circ W^{-1},
\end{equation}
with a dressing operator
\begin{equation}
W(n;t)=1+\sum^\infty_{j=1}w_j(n;t)\Delta^{-j}.
\end{equation}
Further the flow equation
(\ref{1cdKPlaxeq}) is equivalent to the so-called Sato equation,
\begin{equation}\label{tkaction}
\partial_{t_l}W=-(L^{l})_-\circ W.
\end{equation}
Denote the exponential function as following
\begin{equation}\label{expnt}
Exp(n;t,z)=(1+z)^nexp(\sum_{i=1}^\infty t_iz^i)=exp(\sum_{i=1}^\infty(t_i+n\frac{(-1)^{i-1}}{i})z^i),
\end{equation}
then
\begin{equation}\label{}
\Delta Exp(n;t,z)=zExp(n;t,z), \Delta^* Exp^{-1}(n;t,z)=zExp^{-1}(n;t,z).
\end{equation}

There are the wave function $w(n;t,z)$ and the adjoint wave function $w^*(n;t,z)$ for the dKP hierarchy as the following forms:
\begin{equation}\label{wavefun}
w(n;t,z)=W(n;t) Exp(n;t,z)=(1+\frac{w_1(n;t)}{z}+\frac{w_2(n;t)}{z^2}+\cdots)exp(\sum_{i=1}^\infty(t_i+n\frac{(-1)^{i-1}}{i})z^i)
\end{equation}
and
\begin{eqnarray}\label{adjointwavefun}
w^*(n;t,z)&=&(W^{-1}(n-1;t))^* Exp^{-1}(n;t,z)\nonumber\\
&=&(1+\frac{w_1^*(n;t)}{z}+\frac{w_2^*(n;t)}{z^2}+\cdots)exp(\sum_{i=1}^\infty-(t_i+n\frac{(-1)^{i-1}}{i})z^i).
\end{eqnarray}

There also exists a $\tau$ function $\tau(n;t)$ for the dKP hierarchy \cite{Iliev} such that the wave function is expressed by
\begin{eqnarray}\label{}
w(n;t,z)= \frac{\tau(n,t-[z^{-1}])}{\tau(n,t)}Exp(n;t,z),
\end{eqnarray}
and the adjoint wave function is expressed by
\begin{eqnarray}\label{}
w^{*}(n;t,z)= \frac{\tau(n,t+[z^{-1}])}{\tau(n,t)}Exp^{-1}(n;t,z),
\end{eqnarray}
where $[z]=(z,z^2/2,x^3/3,\cdots).$

The difference $\Delta-$Wronskian \cite{LiuS}
\begin{equation}\label{wronskian}
\tau_{\Delta}(n)=W_m^{\Delta}(q_1,q_2,\dots,q_m) =
\left|
\begin{array}{cccc}
 q_1 &q_2 &  \cdots  &q_m \\
  \Delta q_1 & \Delta q_2 &  \cdots  & \Delta q_m \\
    \vdots  &  \vdots  &  \ddots  &  \vdots   \\
   \Delta^{m-1} q_1 &\Delta^{m-1}q_2 &  \cdots  & \Delta^{m-1}q_m
   \end{array}
   \right|
\end{equation} is a $\tau$ function of dKP hierarchy. In this section, we will reduce $\tau_{\Delta}(n)$ in (\ref{wronskian}) to a $\tau$ function of the constrained discrete KP hierarchy.

 Now we  consider a chain of gauge transformation operator of multi-channel  difference type $T_d$ \cite{cst1,hlc2,lmh20131} starting from the initial $m$-component  Lax operator $L^{(0)}=L=L_{+}+ \sum_{i=1}^{m}q_i(t)\Delta^{-1}r_i(t)$,
\begin{equation}\label{succesT2}
L^{[0]}\xrightarrow{T_d^{[1]}(q_1^{[0]})} L^{[1]}\xrightarrow{T_d^{[2]}(q_2^{[1]})} L^{[2]}\rightarrow \dots \rightarrow L^{[n-1]}\xrightarrow{T_d^{[n]}(q_n^{[n-1]})} L^{[n]}.
\end{equation}
Here the index $i$ in the gauge transformation operator $T_d^{[i]}(q_j^{[j-1]})$ $(j>i)$means the $i$-th gauge transformation, and $q_j^{[j-1]}$ (or  $r_j^{[j-1]}$) is transformed by $(j-1)$-steps gauge transformations from $q_j$ (or  $r_j$), $L^{[k]}$ is transformed by $k$-steps gauge transformations from the initial Lax operator $L$.

Now we firstly consider successive gauge transformations in (\ref{succesT2}). We define the operator as
\begin{equation}\label{TMK}
T_{m}=T_d^{[m]}(q_m^{[m-1]})\circ \dots \circ T_d^{[2]}(q_2^{[1]})\circ T_d^{[1]}(q_1^{[0]}),
\end{equation}
in which
\begin{eqnarray}
q_i^{[j]}=T_d^{[j]}(q_j^{[j-1]})\circ \dots \circ T_d^{[2]}(q_2^{[1]})\circ T_d^{[1]}(q_1^{[0]}) q_i,i,j=1,\cdots,m;\\
r_k^{[j]}=((T_d^{[j]})^{-1})^{*}(q_j^{[j-1])})\circ \dots \circ ((T_d^{[2]})^{-1})^{*}(q_2^{[1]})\circ ((T_d^{[1]})^{-1})^{*}(q_1^{[0]}) r_k,j,k=1,\cdots,m.
\end{eqnarray}
It means that $q_i^{[0]}=q_i$, $r_i^{[0]}=r_i$.
We shall find another criterion for the Wronskian entries $f_1,f_2, \cdots, f_n$ leading to cdKP flows.
The following theorem can be easily got from the Ref. \cite{lmh20131}.
\begin{theorem}\label{theor1}
The gauge transformation operator $T_m$ and $T_m^{-1}$ have the following determinant representation:
\begin{eqnarray}\label{TM}
T_m&=&T_d^{[m]}(q_m^{[m-1]})\circ \dots \circ T_d^{[2]}(q_2^{[1]})\circ T_d^{[1]}(q_1^{[0]})\nonumber\\
&=&\frac{1}{W_m^{\Delta}(q_1,q_2,\dots,q_m)}
\left|
\begin{array}{ccccc}
 q_1 &q_2 &  \cdots  &q_m &1\\
  \Delta q_1 & \Delta q_2 &  \cdots  & \Delta q_m &\Delta\\
    \vdots  &  \vdots  &  \vdots &  \ddots  &  \vdots   \\
   \Delta^{m-1} q_1 &\Delta^{m-1}q_2 &  \cdots  & \Delta^{m-1}q_m & \Delta^{m-1}\\
   \Delta^{m} q_1 &\Delta^{m}q_2 &  \cdots  & \Delta^{m}q_m & \Delta^{m}
   \end{array}
   \right|,
\end{eqnarray}
and
\begin{eqnarray}
T_m^{-1}&=&\left|
\begin{array}{ccccc}
 q_1\circ \Delta^{-1} &\Lambda( q_1)&\Lambda(\Delta q_1) &  \cdots  &\Lambda(\Delta^{m-2} q_1)\\
  q_2\circ \Delta^{-1} &\Lambda(q_2) &\Lambda(\Delta q_2) &  \cdots  & \Lambda(\Delta^{m-2} q_2) \\
    \vdots  &  \vdots &  \vdots  &  \ddots  &  \vdots   \\
   q_m\circ \Delta^{-1} &\Lambda( q_m)&\Lambda(\Delta q_m) &  \cdots  & \Lambda(\Delta^{m-2} q_m)
   \end{array}
   \right|
   \frac{(-1)^{m-1}}{\Lambda(W_m^{\Delta}(q_1,q_2,\dots,q_m))}\nonumber\\
 &=& \sum_{i=1}^{m}\phi_i\circ \Delta^{-1}b_i
\end{eqnarray}
with
\begin{equation}
b_i=(-1)^{m+i}\Lambda(\frac{W_m^{\Delta}(q_1,q_2,\dots,q_{i-1},\hat{i},q_{i+1},\dots,q_m)}{W_m^{\Delta}(q_1,q_2,\dots,q_{i-1},q_i,q_{i+1},\dots,q_m)}).
\end{equation} Here $\hat{i}$ means that the column containing $q_i$ is delete from $W_m^{\Delta}(q_1,q_2,\dots,q_{i-1},q_i,q_{i+1},\dots,q_m)$ and the last row is also deleted. Here the determinant of $T_m$ is expanded by the last column and collecting all sub-determinants on the left side of the $\Delta^{i}$ with the action $"\circ"$. And $T_m^{-1}$ is expanded by the first column and all the sub-determinants are on the right side with the action  $"\circ"$.
\end{theorem}

\section{Wronskian solution of constrained discrete KP hierarchy}\label{section3}
Similar to  reference \cite{aratyn1997},
it has \cite{lmh2}
\begin{equation}\label{operatork}
(K \circ q\circ\Delta^{-1}\circ r)_-=K(q)\circ\Delta^{-1}\circ r, (q\circ \Delta^{-1}\circ r\circ K)_-=-q\circ\Delta^{-1}\circ K^*(r),
\end{equation} for a pure-difference operator $K$ and two arbitrary smooth functions ($q, r$).

An important fact is that
there exist two $m$-th order $\Delta$-differential operators
\begin{equation}
A = \Delta^m
 + a_{m-1}\Delta^{m-1}+ \dots+ a_0, B = \Delta^m
 + b_{m-1}\Delta^{m-1}+ \dots+ b_0,
\end{equation}
such that $AL^l$ and $L^lB$ are differential operators. From $(AL^l)_-=0$ and $(L^lB)_-=0$, we get
that $A$ and $B$ annihilate the functions $q^i$ and $r_j$ , i.e., $A(q_1) = \dots = A(q_m) = 0,$ $B^*(r_1) =
\dots = B^*(r_m) = 0$, that implies $q_i \in Ker(A)$, $r_i \in Ker(B^*)$, $i=1,\dots, m$. The dimension of $Ker(A)$ is $m$.

The following theorem provides a criterion for reducing the $\Delta$-Wronskian $\tau$ function in (\ref{wronskian}) of dKP hierarchy to the $\Delta$-Wronskian $\tau$ function   of the cdKP hierarchy defined by (\ref{1cdKPlaxeq}).
\begin{theorem}\label{mainthe}
 The constrained discrete KP hierarchy   has a solution $L=(L)_++\sum_{j=1}^{m}f_j\circ\Delta \circ g_j$ generated by the $\tau$ function $\tau_{\Delta}(n)=W_m^{\Delta}(f_1,\cdots, f_m)\neq0$ satisfies the
$k$-constrained with some suitable functions $q_1,q_2,\cdots,q_M$ and $r_1,r_2,\cdots,r_M$
if and only if
\begin{equation}\label{wronski1}
W_{m+M+1}^{\Delta}(f_1,\cdots,f_m,\Delta^k f_{i_1},\cdots,\Delta^k f_{i_{M+1}})=0
\end{equation}
for all $(M+1)$ indices $1\leq i_1<i_2<\cdots<i_{M+1}\leq m,$ which is equivalent to
{\small \begin{equation}\label{wronski2}
 \mbox{}\hspace{-1.3cm} W_{M+1}^{\Delta}(\frac{W_{m+1}^{\Delta}(f_1,\dots,f_m,\Delta^k f_{i_1})}{W_m^{\Delta}(f_1,\dots,f_m)},\frac{W_{m+1}^{\Delta}(f_1,\dots,f_m,\Delta^k f_{i_2})}{W_m^{\Delta}(f_1,\dots,f_m)},\cdots,\frac{W_{m+1}^{\Delta}(f_1,\dots,f_m,\Delta^k f_{i_{M+1}})}{W_m^{\Delta}(f_1,\dots,f_m)})=0.
\end{equation}}
Here $f_i$ satisfied linear $\Delta-$difference equations
\begin{equation}\label{linearf}
\frac{\partial f_i}{\partial t_k}=\Delta^k f_i,\quad i=1,2,\cdots,m;k=1,2,\cdots.
\end{equation}
\end{theorem}
\begin{proof}
Similar to case of KP hierarchy \cite{JMP1996_Oevel} and $q$-KP hierarchy \cite{sigma06-060}, there has the following Wronskian identity
\begin{eqnarray}
W_{m+M+1}^{\Delta}(\frac{W_{m+1}^{\Delta}(f_1,\dots,f_m,\Delta^k f_{i_1})}{W_m^{\Delta}(f_1,\dots,f_m)},&&\frac{W_{m+1}^{\Delta}(f_1,\dots,f_m,\Delta^k f_{i_2})}{W_m^{\Delta}(f_1,\dots,f_m)},\cdots,
\frac{W_{m+1}^{\Delta}(f_1,\dots,f_m,\Delta^k f_{i_{M+1}})}{W_m^{\Delta}(f_1,\dots,f_m)})\nonumber\\
&&=W_m(f_1,\cdots,f_m,\Delta^k f_{i_1},\cdots,\Delta^k f_{i_{M+1}})
\end{eqnarray}
which provides  the equivalence between (\ref{wronski1}) and (\ref{wronski2}).

Using Theorem \ref{theor1} and the relation of operator identities in (\ref{operatork}), one finds
\begin{equation}\label{lkconstrian}
(L^k)_-=(W \circ \Delta^k \circ W^{-1})_-= \sum_{j=1}^{m}W (\Delta^k  f_j )\Delta^{-1}g_j
\end{equation}
 As it was pointed out in the beginning of this section,  there exists an $m$-th order differential operator $A$ such that $AL^k$ is a difference operator. Application to $W(f_j)=0$ yields
 \begin{equation}
0=AL^k(W(f_j))=AW\circ\Delta^k (f_j)= AW (\partial_{t_k}(f_j)).
\end{equation}
So,
 \begin{equation}
W(\partial_{t_k}(f_j))=\frac{W_m^{\Delta}(f_1,\cdots, f_m,\Delta^k f_j)}{W_m^{\Delta}(f_1,\cdots, f_m)}\in Ker (A).
\end{equation}
Since the kernel of $A$ has dimension $m$, at most $m$ of these functions $\Delta^k f_j$ can be linearly independent. So, (\ref{wronski2}) is deduced.

Conversely, if (\ref{wronski2}) holds, then  there exists one $M$-component of cdKP
$(M < m)$ constrained from (\ref{lkconstrian}). The equation (\ref{wronski2}) implies that at most $M$ of functions $W(\Delta^k (f_j))$
are linearly independent, here $f_j$ satisfy (\ref{linearf}). Then we can find suitable $M$ functions
${q_1, q_2, \dots , q_M}$, which are linearly independent, to express functions $W(\Delta^k (f_j))$  as
 \begin{equation}
W(\partial_{t_k}(f_j))=\frac{W_m^{\Delta}(f_1,\cdots, f_m,\Delta^k f_j)}{W_m^{\Delta}(f_1,\cdots, f_m)}=\sum_{i=1}^{M}c_{ij}q_i, j=1,\dots, m.
\end{equation}
with some constant $c_{ij}$. Taking it back into the (\ref{lkconstrian}), it becomes
 \begin{equation}
(L^k)_-= \sum_{j=1}^{m}(\sum_{i=1}^{M}c_{ij}q_i)\circ\Delta^{-1}\circ g_j=\sum_{i=1}^{M}q_i\circ\Delta^{-1}\circ(\sum_{j=1}^{m}c_{ij}g_j)=\sum_{i=1}^{M}q_i\circ\Delta^{-1}\circ r_i,
\end{equation} then a $m$-component cdKP hierarchy is reduced to a $M$-component cdKP hierarchy.
\end{proof}

\textbf{Remark:} This theorem is a difference version of the corresponding theorem of the Ref. \cite{JMP1996_Oevel}.

The Wronskian solution of the cdKP hierarchy can be got by the Theorem \ref{mainthe} under the gauge transformation. If the initial Lax operator of the constrained discrete KP hierarchy is a "free" operator $\Delta$, then
$L= \Delta$ means that the initial $\tau$ function $\tau_\Delta$ is $1$.

\section{Example of reducing dKP hierarchy to cdKP hierarchy}\label{example}

In this section, we  use the method in Theorem \ref{mainthe} to find the solution of the  multi-component cdKP hierarchy.
We discuss the cdKP hierarchy generated by $T_i\mid_{i=2}$, possesses a $\tau$ function
\begin{equation}
\tau_{\Delta}^{(2)}=W_2^{\Delta}(f_1,f_2)=
\left|
\begin{array}{cc}
f_1 &f_2 \\
 \Delta f_1 & \Delta f_2
\end{array}
   \right|=f_1\circ\Delta f_2- f_2\circ\Delta f_1,
\end{equation}
with
\begin{equation}
f_1=f_{11}(z_1,n,t)+f_{12}(z,n,t), f_2=f_{21}(z_2,n,t)+f_{22}(z_3,n,t).
\end{equation}
Here
\begin{eqnarray*}
f_{11}(z_1,n,t)=(1+z_1)^ne^{\xi_1},f_{12}(z,n,t)=(1+z)^ne^{\xi},\\
 f_{21}(z_2,n,t)=(1+z_2)^ne^{\xi_2},f_{22}(z_3,n,t)=(1+z_3)^ne^{\xi_3},
\end{eqnarray*}
where $\xi_i=c_i+z_it_1+z_i^2t_2+z_i^3t_3,i=1,2,3$ and $\xi=d+zt_1+z^2t_2+z^3t_3,$ $c_i$ and $d$ are arbitrary constants.
These functions $f_1$ and $f_2$ satisfy  linear equations (\ref{linearf})  for $k=1,2,3$.
By (\ref{linearf}), the cdKP hierarchy generated by $T_2$ is in the form of
\begin{eqnarray}
L^l&=&(L^l)_++(T_2(\Delta^kf_1))\circ\Delta^{-1}\circ g_1+(T_2(\Delta^kf_2))\circ\Delta^{-1}\circ g_2\label{t2constrain}\\
&\stackrel{constraint}{====}& (L^l)_++q_1\circ\Delta^{-1}\circ r_1,\label{t2constrain2}
\end{eqnarray}
where $q_1,r_1$ are undetermined, which can be expressed by $f_1$ and $f_2$ as follows.
$\tau_{\Delta}^{(2)}$ possesses a form as
\begin{eqnarray}\label{t2}
\tau_{\Delta}^{(2)}&=&W_2^{\Delta}(f_1,f_2)\nonumber\\
&=&(z_2-z_1)(1+z_1)^n(1+z_2)^ne^{\xi_1+\xi_2}+(z_2-z)(1+z)^n(1+z_2)^ne^{\xi+\xi_2}\nonumber\\
&&+(z_3-z_1)(1+z_1)^n(1+z_3)^ne^{\xi_1+\xi_3}+(z_3-z)(1+z)^n(1+z_3)^ne^{\xi+\xi_3}.
\end{eqnarray}

According to (\ref{wronski1}) in Theorem \ref{mainthe}, the restriction for $f_1$ and $f_2$ to reduce (\ref{t2constrain}) to (\ref{t2constrain2}) is given by
\begin{eqnarray}\label{w2}
0&=&W_4^{\Delta}(f_1,f_2,f_1^{(k)},f_2^{(k)})\nonumber\\
&=&(z^k-z_1^k)(z_2^k-z_3^k)V(z_1,z_2,z_3,z)F(n;z_i,t),
\end{eqnarray}
with Vandermonde determinant
\begin{equation}
V(z_1,z_2,z_3,z)=\left|
\begin{array}{cccc}
1&1&1&1\\
z_1&z_2&z_3&z\\
z_1^2&z_2^2&z_3^2&z^2\\
z_1^3&z_2^3&z_3^3&z^3
\end{array}
   \right|,
\end{equation}
and
\begin{equation}
F(n;z_i,t)=(1+z_1)^n(1+z_2)^n(1+z_3)^n(1+z)^ne^{\xi_1+\xi_2+\xi_3+\xi}.
\end{equation}

Obviously, $f_1$ and $f_2$  satisfy  (\ref{w2}) by setting $z=z_2$ and $d=c_2$. Then the $\tau$ function of a single component $k$-constrained cdKP hierarchy defined by
\begin{eqnarray}
\tau_{cdKP}^{\Delta}&=&(z_2-z_1)(1+z_1)^n(1+z_2)^ne^{\xi_1+\xi_2}+(z_3-z_1)(1+z_1)^n(1+z_3)^ne^{\xi_1+\xi_3} \nonumber\\
&&+(z_3-z_2)(1+z_2)^n(1+z_3)^ne^{\xi_2+\xi_3},
\end{eqnarray}
which is deduced by (\ref{t2}) with $\xi_2=\xi$. It means that we indeed reduced the $\tau$ function $\tau_{\Delta}^{(2)}$ in (\ref{t2}) of the dKP hierarchy to the $\tau$ function $\tau_{cdKP}^{\Delta}$ of the $1$-component cdKP hierarchy.

  We would like to get the explicit forms  of $q_1$ and $r_1$ of cdKP hierarchy in (\ref{t2constrain2}).
With the determinant representation of $T_2$ and $T_2^{-1}$, one can have
\begin{subequations}\label{f1f2}
\begin{align}
f_1^{\Delta}&\triangleq T_2(\Delta^kf_1)=\frac{(z_1^k-z_2^k)V(z_1,z_2,z_3)(1+z_1)^n(1+z_2)^n(1+z_3)^ne^{\xi_1+\xi_2+\xi_3}}{\tau_{cdKP}^{\Delta}}&\label{f1f2:1A}\\
f_2^{\Delta}&\triangleq T_2(\Delta^k f_2)=\frac{(z_3^k-z_2^k)V(z_1,z_2,z_3)(1+z_1)^n(1+z_2)^n(1+z_3)^ne^{\xi_1+\xi_2+\xi_3}}{\tau_{cdKP}^{\Delta}}\label{f1f2:1B}\\
g_1^{\Delta}&\triangleq (T_2^{-1})^*(\Delta^k g_1)=-\Lambda(\frac{f_2}{W_2^{\Delta}(f_1,f_2)})=-\Lambda(\frac{(1+z_2)^ne^{\xi_2}+(1+z_3)^ne^{\xi_3}}{\tau_{cdKP}^{\Delta}})\label{f1f2:1C}\\
g_2^{\Delta}&\triangleq (T_2^{-1})^*(\Delta^k g_1)=\Lambda(\frac{f_1}{W_2^{\Delta}(f_1,f_2)})=\Lambda(\frac{(1+z_1)^ne^{\xi_1}+(1+z_2)^ne^{\xi_2}}{\tau_{cdKP}^{\Delta}})\label{f1f2:1D}
\end{align}
\end{subequations}
with the Vandermonde determinant
\begin{equation}
V(z_1,z_2,z_3)=\left|
\begin{array}{ccc}
1&1&1\\
z_1&z_2&z_3\\
z_1^2&z_2^2&z_3^2
\end{array}
   \right|.
\end{equation}
It is clearly
\begin{equation*}
(z_3^k-z_2^k)f_1^{\Delta}=(z_1^k-z_2^k)f_2^{\Delta}.
\end{equation*}

So the $q_1$ in (\ref{t2constrain2}) is
\begin{eqnarray}
\mbox{\hspace{-0.3cm}}&q_1\mbox{\hspace{-0.3cm}}&\triangleq(z_3^k-z_2^k)f_1^{\Delta}=(z_1^k-z_2^k)f_2^{\Delta}\nonumber\\
\mbox{\hspace{-0.3cm}}&=&\mbox{\hspace{-0.3cm}}\frac{(z_3^k-z_2^k)(z_1^k-z_2^k)(z_2-z_1)(z_3-z_1)(z_3-z_2)(1+z_1)^n(1+z_2)^n(1+z_3)^ne^{\xi_1+\xi_2+\xi_3}}{\tau_{cdKP}^{\Delta}}.
\end{eqnarray}
And the (\ref{t2constrain}) is reduced to
\begin{eqnarray}
(L^l)_-&=&f_1^{\Delta}\circ\Delta^{-1}\circ g_1^{\Delta}+f_2^{\Delta}\circ\Delta^{-1}\circ g_2^{\Delta}\nonumber\\
&=& (z_3^k-z_2^k)f_1^{\Delta}\circ\Delta^{-1}\circ \frac{g_1^{\Delta}}{(z_3^k-z_2^k)}+(z_1^k-z_2^k)f_2^{\Delta}\circ\Delta^{-1}\circ \frac{g_2^{\Delta}}{z_1^k-z_2^k}\nonumber\\
&=& q_1\circ\Delta^{-1}\circ r_1,
\end{eqnarray}
where
\begin{eqnarray}
r_1&\triangleq&\frac{g_1^{\Delta}}{(z_3^k-z_2^k)}+\frac{g_2^{\Delta}}{z_1^k-z_2^k}\nonumber\\
&=&\Lambda(\frac{(z_3^k-z_2^k)(1+z_1)^ne^{\xi_1}+(z_3^k-z_1^k)(1+z_2)^ne^{\xi_2}+(z_2^k-z_1^k)(1+z_3)^ne^{\xi_3}}{(z_3^k-z_2^k)(z_1^k-z_2^k)\tau_{cdKP}^{\Delta}}).
\end{eqnarray}

For simplicity, denote $t_1 = x; t_2 = y$.
In particular, choosing $z_1=z$, $z_2=0$, $z_3=-z$, $c_1=c$, $c_2=0$, $c_3=-c$ then $\xi_2=0$, $\xi_3=-\xi_1$ and
\begin{equation}\label{q1}
q_1=\frac{(-1)^{k+2}z^{2k+2}(1+z)^n(1-z)^ne^{2z^2y}}{(1+z)^n(1-z)^n+\frac{(1+z)^ne^{\eta}+(1-z)^ne^{-\eta}}{2}e^{z^2y}}.
\end{equation}
Base on above choice,
\begin{equation}\label{r1k1}
r_1=
-\frac{1}{z^{k+1}}\Lambda(\frac{\frac{[(1+z)^ne^{\eta}+(1-z)^ne^{-\eta}]}{2}+e^{-z^2y}}{\frac{[(1+z)^ne^{\eta}+(1-z)^ne^{-\eta}]}{2}+(1-z^2)^ne^{z^2y}})
\end{equation}and if $k$ is odd, and
\begin{equation}\label{r1k2}
r_1=
-\frac{1}{z^{k+1}}\Lambda(\frac{\frac{(1+z)^ne^{\eta}-(1-z)^ne^{-\eta}}{2}}{\frac{(1+z)^ne^{\eta}+(1-z)^ne^{-\eta}}{2}+(1-z^2)^ne^{z^2y}})
\end{equation}
if $k$ is even. Here $\eta \triangleq c+zx+z^3t_3.$

So  the dynamical variable $u_1=q_1\Lambda^{-1}(r_1)$ of the Lax operator $L$ of the cdKP hierarchy
\begin{equation}\label{u1}
u_1=2z^{k+1}\frac{(1-(-1)^k)+e^{-(c+zx+z^3t_3)+z^2y}(1-z)^{n}
-e^{(c+zx+z^3t_3)+z^2y}(-1)^k(1+z)^{n}}{(1-z^2)^n(\frac{e^{c+zx+z^3t_3}}{(1-z)^n}+2e^{z^2y}+\frac{e^{-(c+zx+z^3t_3)}}{(1+z)^n})^2}.
\end{equation}
An example is
\begin{equation}
u_1=2 z^{k+1}\frac{(1-(-1)^k)+e^{-zx+z^2y}(1-z)^{n}
-e^{zx+z^2y}(-1)^k(1+z)^{n}}{(1-z^2)^n(\frac{e^{zx}}{(1-z)^n}+2e^{z^2y}+\frac{e^{-zx}}{(1+z)^n})^2},
\end{equation} by  setting $t_3=0, c=0.$
For this case
\begin{eqnarray}
u_1=
2 z^{k+1}\frac{2+e^{-zx+z^2y}(1-z)^{n}
+e^{zx+z^2y}(1+z)^{n}}{(1-z^2)^n(\frac{e^{zx}}{(1-z)^n}+2e^{z^2y}+\frac{e^{-zx}}{(1+z)^n})^2},\text{$k$ is odd},\label{u1k1}\\
u_1=2 z^{k+1}\frac{e^{-zx+z^2y}(1-z)^{n}
-e^{zx+z^2y}(1+z)^{n}}{(1-z^2)^n(\frac{e^{zx}}{(1-z)^n}+2e^{z^2y}+\frac{e^{-zx}}{(1+z)^n})^2},\text{$k$ is even.}\label{u1k2}
\end{eqnarray}

\textbf{Remark:} Actually, (\ref{w2}) also can be satisfied by other two choices $z=z_1$ or $z=z_3$. But $u_1$ (\ref{u1}) will be only one-soliton solution  because the $f_1^{\Delta}=0$ in (\ref{f1f2:1A}) or $f_2^{\Delta}=0$ in (\ref{f1f2:1B}) separately.

The graph of $q_1=q_1(x,y,n),r_1=r_1(x,y,n),u_1=u_1(x,y,n)$ were plotted in below for fixed $k=1,2$.
We shall discuss the function of the gauge transformation for the the cdKP hierarchy to emphasize  two sides about the discrete variable $n$ of it and the times variable $k$ of the gauge transformation of it.
 The profile of $q_1,r_1,u_1$ are plotted according to the value of discrete variable $n$ from $0$ to $2$ and the value of time $k$ of the gauge transformations  from $1$ to $2$.
The five conditions of the profile of $q_1,r_1,u_1$ are  $\{n=0,k=1\}$,$\{n=1,k=1\}$, $\{n=2,k=1\}$, $\{n=0,k=2\}$, $\{n=1,k=2\}$ and $\{n=2,k=2\}$ as following.

(1).The profile of $q_1$  are plotted with $k=1,n=0,1,2$ in Figure \ref{q1n0k1}, Figure \ref{q2n1k1} and Figure \ref{q2n2k1} respectively.

(2).The profile of $r_1$ are plotted with $k=1,n=0,1,2$ in Figure \ref{r0n0k1}, Figure \ref{r2n1k1} and Figure \ref{r2n2k1} respectively.

(3).The  Y-type soliton profile  of $u_1$ are plotted with $k=1,n=0,1,2$ in Figure \ref{u2n0k1} , Figure \ref{u2n1k1} and Figure \ref{u2n2k1} respectively.

(4).The bright-dark soliton profile of $u_1$ are plotted with $k=2$ and $n=0,1,2$  in Figure \ref{u2n0k2}, Figure \ref{u2n1k2} and  Figure \ref{u2n2k2}  respectively.

 From the graphs of the solution of cdKP hierarchy, it can be found that:

(1) The profile of the solution $q_1$ of the cdKP hierarchy is decreasing to the one of the  classical KP hierarchy in Ref. \cite{JMP1996_Oevel}  when $n\to 0$  see Figure \ref{q1n0k1} ($k=1$). For $r_1,u_1$ of the cdKP hierarchy, the profile of its are also decreasing the analogues of the classical KP hierarchy (see  Figure \ref{r0n0k1} and Figure \ref{u2n0k1}).

(2)When  the times $k$ of gauge transformation is an odd number,  the profiles of $u_1$ become the Y-type soliton, see Figure \ref{u2n0k1}, Figure \ref{u2n1k1} and Figure \ref{u2n2k1}.

(3)When the times $k$  of gauge transformation is an even number,  the profiles of $u_1$ become bright-dark soliton, see Figure \ref{u2n0k2}, Figure \ref{u2n1k2} and Figure \ref{u2n2k2}.

For the end of showing more detail about dependence of $u_1,q_1$ on $n$, it is necessary to define $n$-effect quantity $\Delta u_1(z,x,y,n)=u_1(z,x,y,n)-u_1(z,x,y,n=0)=u_1(n)-u_1(0)$ for fixed $z=0.5$. Figure \ref{udd} are plotted for the $\Delta u_1(z,x,y,n)$ where $n=1,2,3$ respectively, which shows  the dependence of $u_1$ on $n$. It was obviously they are decreasing to almost zero when $n$  goes from $3$ to $1$ with fixed $z=0.5$. They also demonstrate that discretization of the cdKP hierarchy keeps the profile of the soliton though it has  discrete variable $n$. These figures give us again an opportunity to observe the role of discrete variable $n$  in the Wronskian solution of the cdKP hierarchy.

\section{Conclusions}\label{conclusion}
In this paper, the Wronskian solutions of the equation in the cdKP hierarchy have been given by means of the multi-channel  gauge transformation. Based on the results  of our previous papers \cite{lmh20131,lmh2}, Theorem \ref{mainthe} provides a necessary  and sufficient condition of the $k$-constrained discrete KP hierarchy with $m$ components. As an example, the reduction  from $2$-cdKP hierarchy to $1$-cdKP hierarchy is presented. It can be found that the profiles of solution $u_1$ of cdKP hierarchy can be the Y-type solition by the odd number times gauge transformation, but the solution of cdKP hierarchy  becomes bright-dark solition by even times gauge transformation. From these profiles, it can be find that the solution $u_1$ of the cdKP hierarchy is decreasing to the analogues of the  classical KP hierarchy  when $n\to 0$.

{\bf Acknowledgments} {\noindent
This work is supported by
the National Natural Science Foundation of China  under Grant Nos.11271210 and 11201251, K.C.Wong Magna Fund in
Ningbo University, the Natural Science Foundation of Zhejiang Province under Grant No. LY12A01007 and Science Fund of Ningbo
University (No.XYL14028). One of the authors (MH) is  supported by Erasmus Mundus Action 2 EXPERTS III
 and would like to thank Prof. Antoine Van Proeyen for many helps.}


\begin{figure}[htb]
\centering
\raisebox{0.8in}{}\includegraphics[scale=0.8]{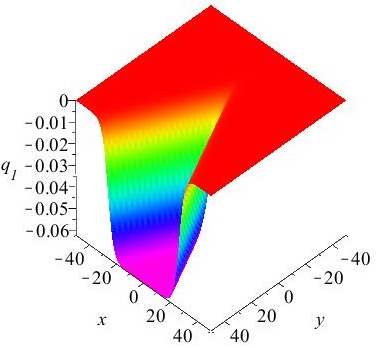}
\hskip 0.03cm
\raisebox{0.8in}{}\includegraphics[scale=0.4]{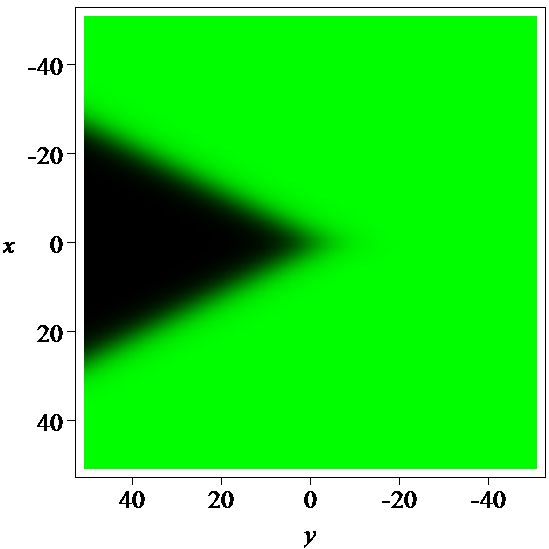}
\caption{
The profile of the solution $q_1$(\ref{q1}) of equation (\ref{xuedinger}) (left) and its  density plot (right)with $ c=0, t_3=0, z=0.5, k=1$ and $n=0$ .}\label{q1n0k1}
\end{figure}

\begin{figure}[htb]
\setlength{\unitlength}{0.1cm}
\raisebox{0.85in}{}\includegraphics[scale=0.6]{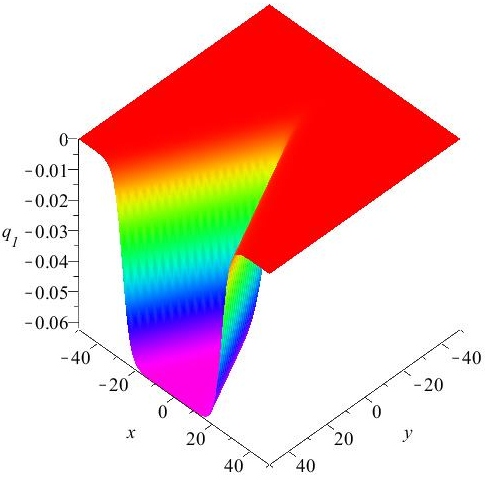}
\setlength{\unitlength}{0.1cm}
\raisebox{0.85in}{}\includegraphics[scale=0.4]{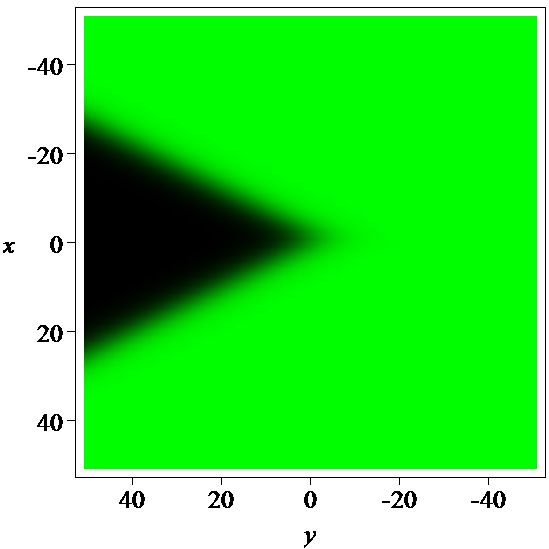}
\caption{The profile of the solution of $q_1$(\ref{q1}) of equation (\ref{xuedinger}) (left) and its  density plot (right) with parameters $ c=0, t_3=0, z=0.5, k=1$ and $n=1$.}\label{q2n1k1}
\end{figure}

\begin{figure}[htb]
\setlength{\unitlength}{0.1cm}
\raisebox{0.85in}{}\includegraphics[scale=0.6]{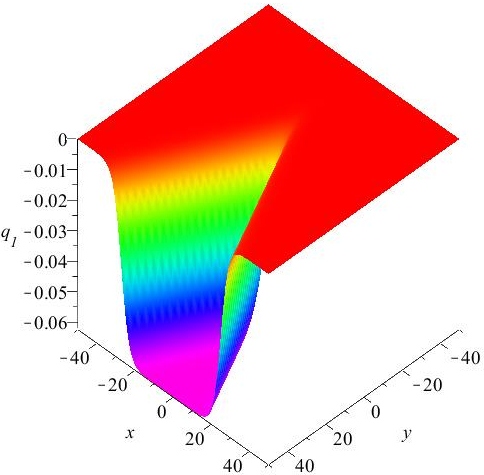}
\setlength{\unitlength}{0.1cm}
\raisebox{0.85in}{}\includegraphics[scale=0.4]{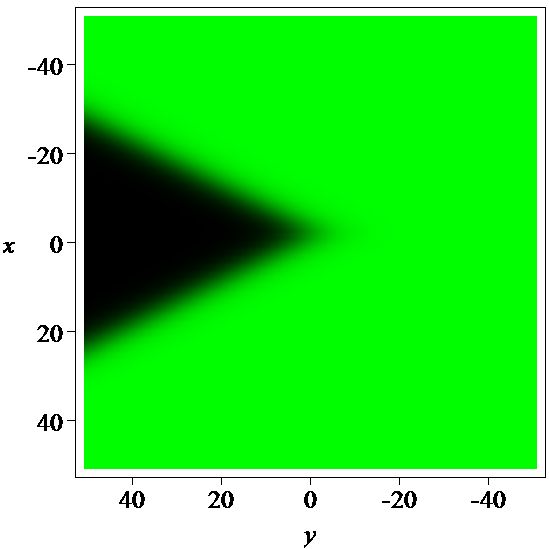}
\caption{The profile of the solution of $q_1$(\ref{q1}) of equation (\ref{xuedinger}) (left) and its  density plot (right) with parameters $ c=0, t_3=0, z=0.5, k=1$ and $n=2$.}\label{q2n2k1}
\end{figure}

\begin{figure}[htb]
\centering
\raisebox{0.8in}{}\includegraphics[scale=0.65]{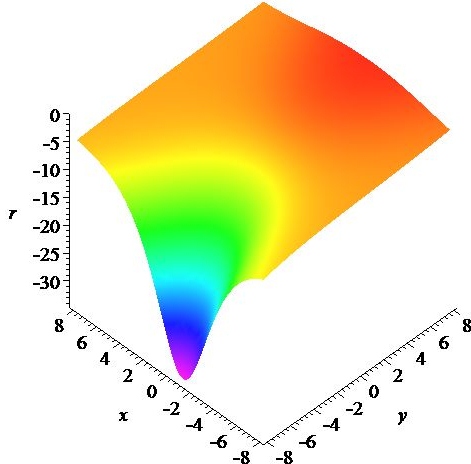}
\hskip 0.03cm
\raisebox{0.8in}{}\includegraphics[scale=0.4]{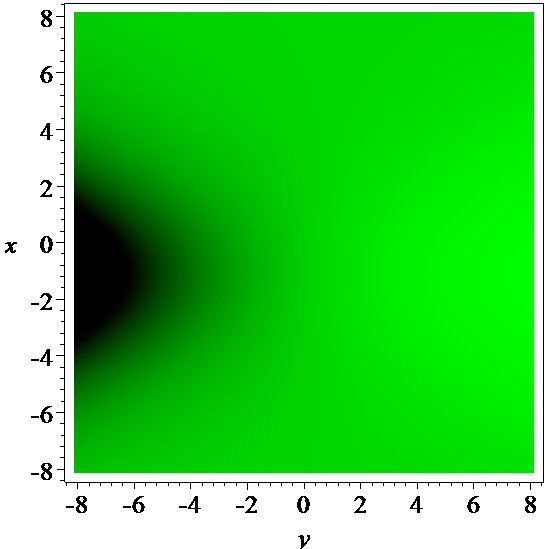}
\caption{
The profile of the solution  $ r_1$ (\ref{r1k1}) (left)  of equation (\ref{xuedinger}) and its  density plot (right) with parameters $ c=0, t_3=0, z=0.5, k=1$ and $n=0$. The vertical axis $r$ denotes the $r_1$.}\label{r0n0k1}
\end{figure}

\begin{figure}[htb]
\setlength{\unitlength}{0.1cm}
\raisebox{0.85in}{}\includegraphics[scale=0.62]{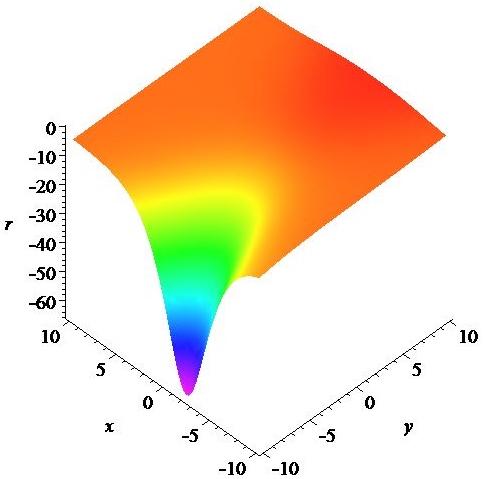}
\setlength{\unitlength}{0.1cm}
\raisebox{0.85in}{}\includegraphics[scale=0.4]{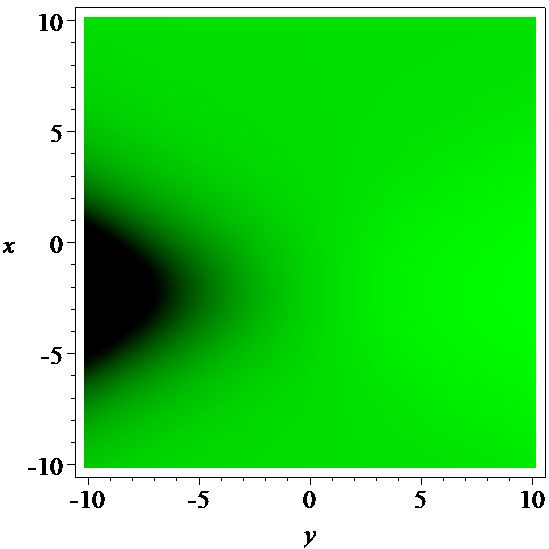}
\caption{The profile of $r_1$ (\ref{r1k1}) (left) and its  density plot (right) with parameters $ c=0, t_3=0, z=0.5, k=1 $ and $n=1$.
The vertical axis $r$ denotes the $r_1$.}\label{r2n1k1}
\end{figure}

\begin{figure}[htb]
\setlength{\unitlength}{0.1cm}
\raisebox{0.85in}{}\includegraphics[scale=0.62]{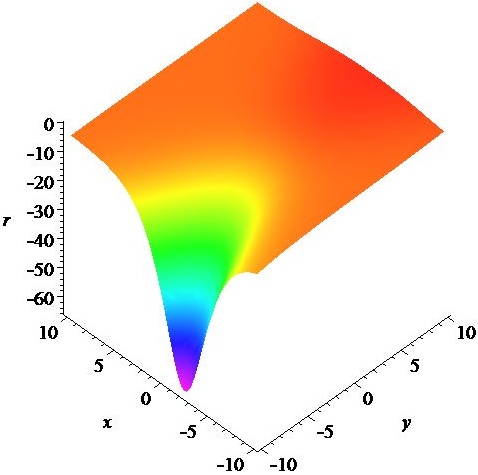}
\setlength{\unitlength}{0.1cm}
\raisebox{0.85in}{}\includegraphics[scale=0.4]{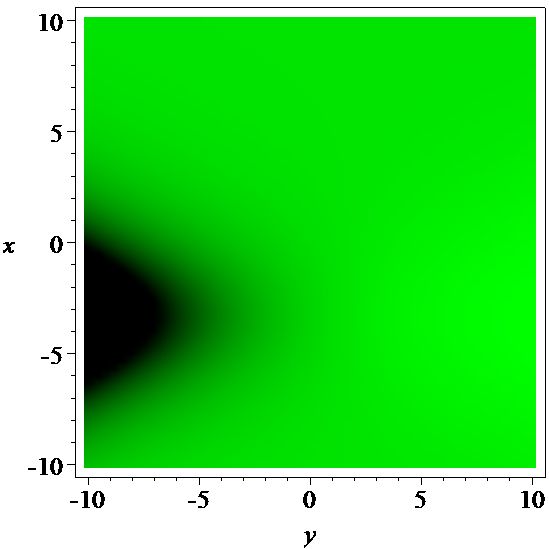}
\caption{The profile of $r_1$ (\ref{r1k1}) (left) and its  density plot (right) with parameters $ c=0, t_3=0, z=0.5, k=1 $ and $n=2$.
The vertical axis $r$ denotes the $r_1$. }\label{r2n2k1}
\end{figure}

\begin{figure}[htb]
\centering
\raisebox{0.8in}{}\includegraphics[scale=0.67]{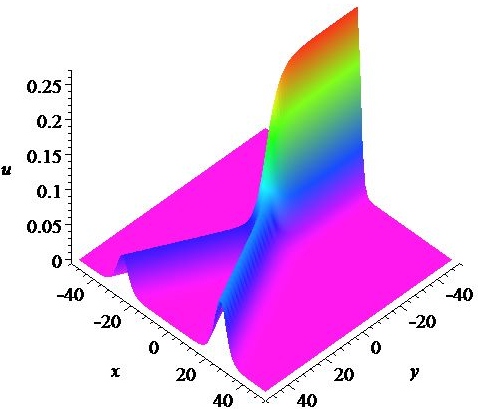}
\hskip 0.03cm
\raisebox{0.8in}{}\includegraphics[scale=0.4]{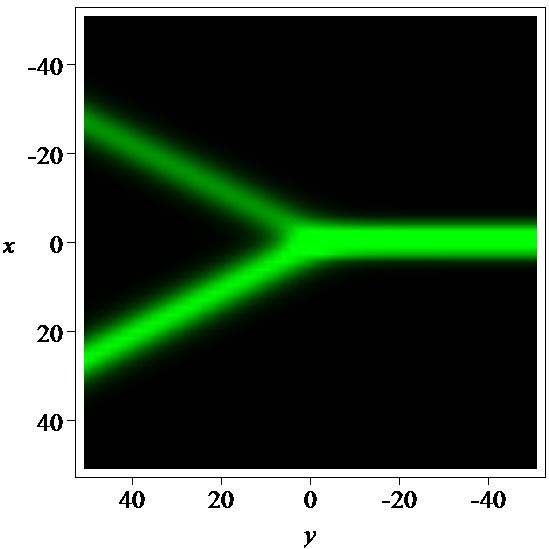}
\caption{The profile of type Y soliton of $u_1$ (\ref{u1k1}) (left) and its  density plot (right) with parameters $ c=0, t_3=0, z=0.5, k=1$ and  $n=0$. The vertical axis $u$ denotes the $u_1$.}\label{u2n0k1}
\end{figure}

\begin{figure}[htb]
\setlength{\unitlength}{0.1cm}
\raisebox{0.80in}{}\includegraphics[scale=0.70]{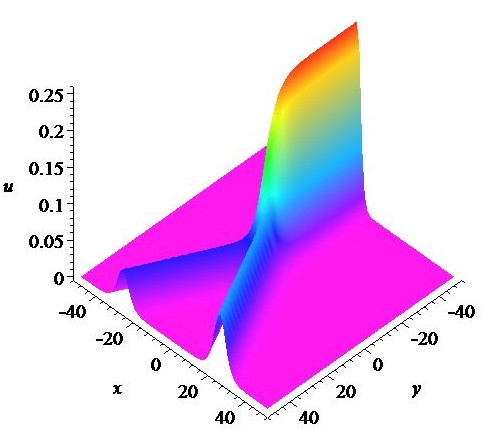}
\setlength{\unitlength}{0.1cm}
\raisebox{0.80in}{}\includegraphics[scale=0.4]{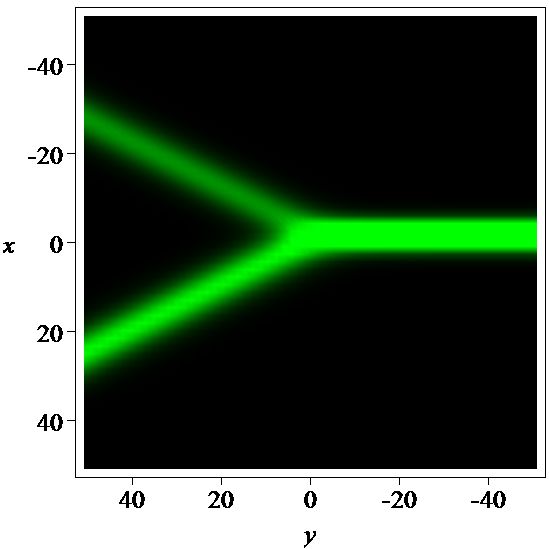}
\caption{The profile of type Y soliton of $u_1$ (\ref{u1k1}) (left) and its  density plot (right) with parameters $ c=0, t_3=0, z=0.5, k=1$  and $n=1$.
The vertical axis $u$ denotes the $u_1$.}\label{u2n1k1}
\end{figure}

\begin{figure}[htb]
\setlength{\unitlength}{0.1cm}
\raisebox{0.85in}{}\includegraphics[scale=0.45]{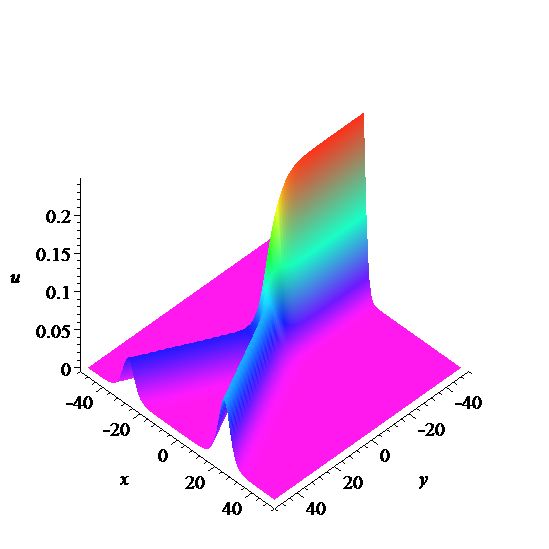}
\raisebox{0.85in}{}\includegraphics[scale=0.40]{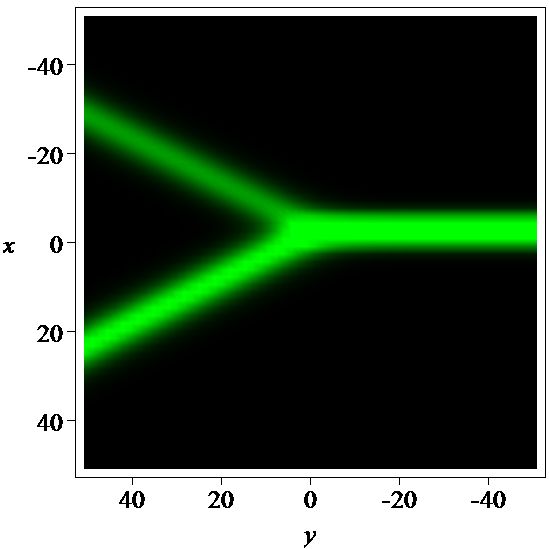}
\caption{The profile of type Y soliton of $u_1$ (\ref{u1k1}) (left) and its  density plot (right) with parameters $ c=0, t_3=0, z=0.5, k=1$  and $n=2$. The vertical axis $u$ denotes the $u_1$.}\label{u2n2k1}
\end{figure}

\begin{figure}[htb]
\setlength{\unitlength}{0.1cm}
\raisebox{0.85in}{}\includegraphics[scale=0.68]{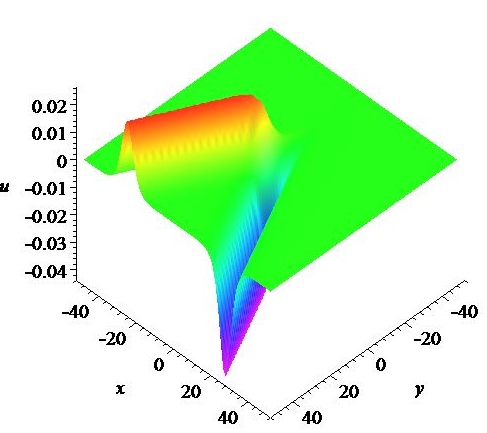}
\setlength{\unitlength}{0.1cm}
\raisebox{0.4in}{}\includegraphics[scale=0.4]{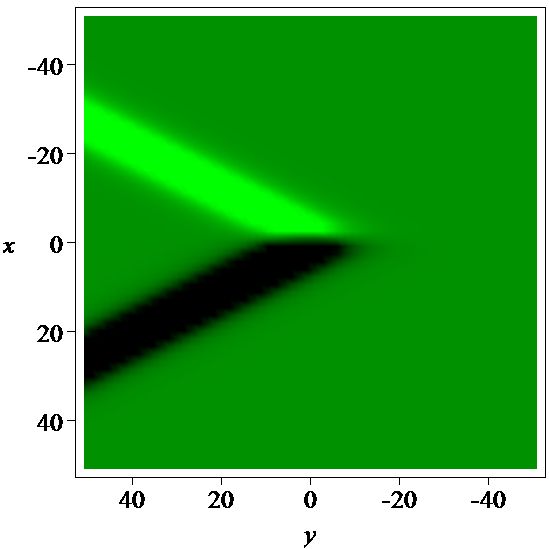}
\caption{The profile of bright-dark type soliton of $u_1$ (\ref{u1k2}) (left) and its  density plot (right) with parameters $c=0, t_3=0, z=0.5, k=2$  and $n=0$ . The vertical axis $u$ denotes the $u_1$.}\label{u2n0k2}
\end{figure}

\begin{figure}[htb]
\setlength{\unitlength}{0.1cm}
\raisebox{0.85in}{}\includegraphics[scale=0.68]{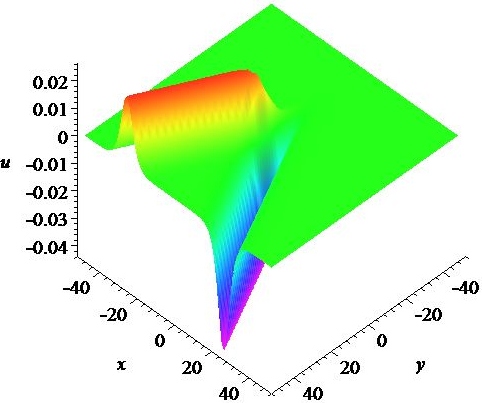}
\setlength{\unitlength}{0.1cm}
\raisebox{0.4in}{}\includegraphics[scale=0.4]{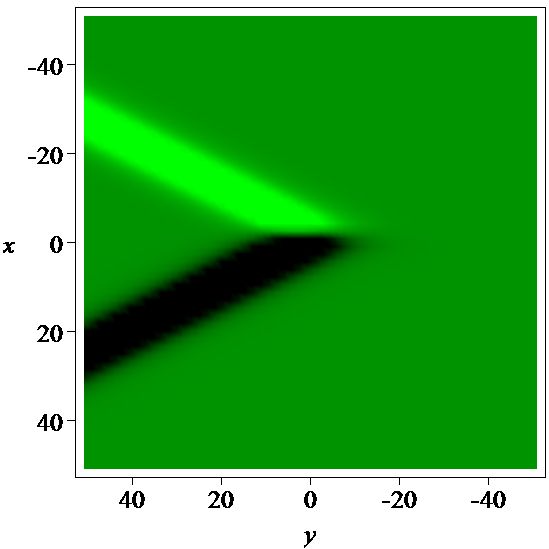}
\caption{The profile of bright-dark type soliton of $u_1$ (\ref{u1k2}) (left) and its  density plot (right) with parameters $c=0, t_3=0, z=0.5, k=2$  and $n=1$ . The vertical axis $u$ denotes the $u_1$.}\label{u2n1k2}
\end{figure}

\begin{figure}[htb]
\setlength{\unitlength}{0.1cm}
\raisebox{0.85in}{}\includegraphics[scale=0.70]{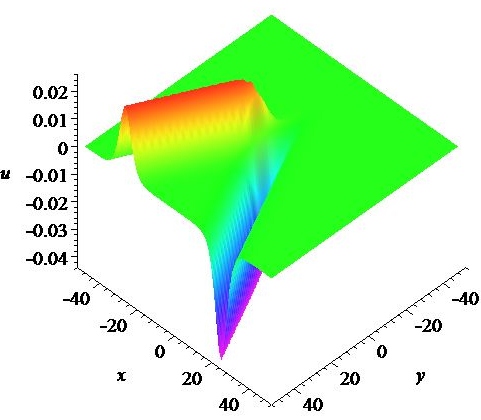}
\setlength{\unitlength}{0.1cm}
\raisebox{0.85in}{}\includegraphics[scale=0.4]{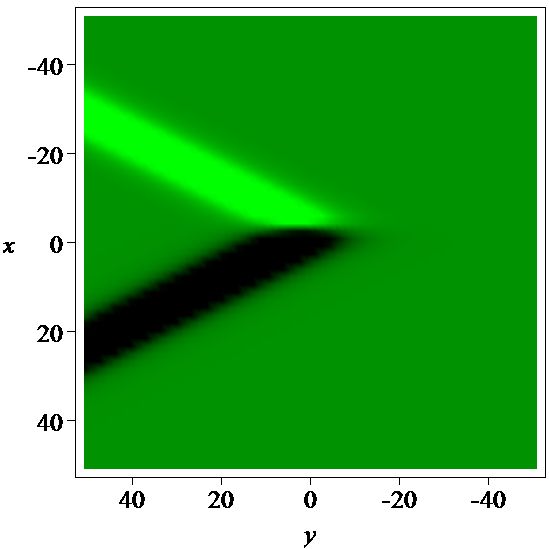}
\caption{
The profile of bright-dark soliton of $u_1$ (\ref{u1k2}) (left) and its  density plot (right) with parameters $c=0, t_3=0, z=0.5, k=2$  and $n=2$. The vertical axis $u$ denotes the $u_1$.}\label{u2n2k2}
\end{figure}

\begin{figure}[htb]
\centering
\raisebox{0.8in}{(a)}\includegraphics[scale=0.6]{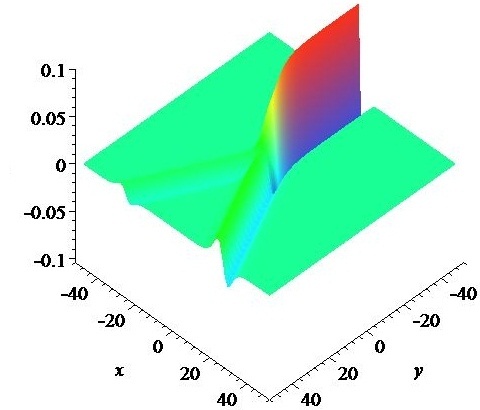}
\hskip 0.03cm
\raisebox{0.8in}{(b)}\includegraphics[scale=0.6]{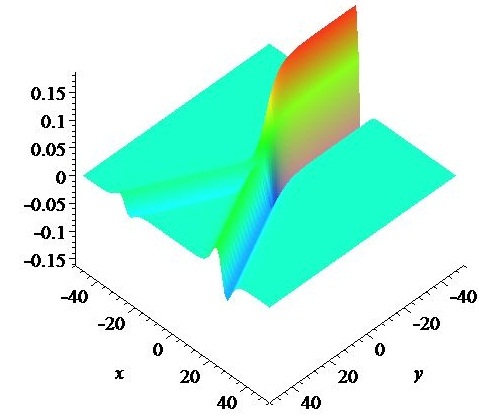}
\hskip 0.03cm
\raisebox{0.8in}{(c)}\includegraphics[scale=0.65]{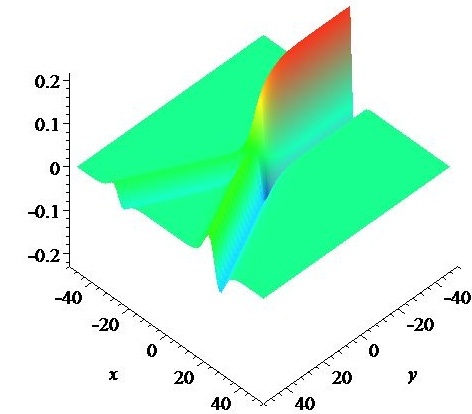}
\caption{
 $\Delta u_1(z,x,y,n)$ with $ c=0, t_3=0, z=0.5, k=1$ and $n=1$ in (a), $2$ in (b) and $3$ in (c).}\label{udd}
\end{figure}

\end{document}